\RequirePackage{fix-cm}
\documentclass[smallextended]{svjour3}       % onecolumn (second format)
\smartqed  % flush right qed marks, e.g. at end of proof
\usepackage{graphicx}
\usepackage{mathrsfs}
\usepackage{amsfonts}
\usepackage{graphicx}
\usepackage{amssymb}
\usepackage{amsmath}
\usepackage{cases}
\usepackage{caption}
\usepackage{hyperref}

\newtheorem{Thm}{Theorem}

\newtheorem{Lem}[Thm]{Lemma}

\newtheorem{Exp}[Thm]{Example}

\newcommand{\F}{{\mathbb  F}}
\newcommand{\wt}{{\mathrm{wt}}}
\newcommand{\tr}{{\mathrm{Tr}}}
\newcommand{\C}{{\mathcal C}}

\makeatletter
 \@addtoreset{equation}{section}
 
\makeatother
\journalname{}
\begin{document}

\title{A class of three-weight and five-weight linear codes
\thanks{ This research is supported by a National Key Basic Research Project of China (2011CB302400), National Science Foundation of China (61379139)
and  the ``Strategic Priority Research Program" of the Chinese Academy of Sciences, Grant No. XDA06010701 and Foundation of NSSFC(No.13CTJ006).}}
%\subtitle{Do you have a subtitle?\\ If so, write it here}

%\titlerunning{Short form of title}        % if too long for running head

\author{Fei Li$^{1}$, Qiuyan Wang$^{*,2,3}$, Dongdai Lin$^{2}$}

%\authorrunning{Short form of author list} % if too long for running head

\institute{F. Li \at $1$. School of Statistics and Applied Mathematics,   Anhui University of Finance and Economics, Bengbu City,  Anhui Province,  233030, China \\
\email{cczxlf@163.com}\\
             Q. Wang,  \at
             $2$. State Key Laboratory of Information Security, Institute of Information Engineering,
             Chinese Academy of Sciences,
              Beijing, 100195,  China\\
             $3$.  University of Chinese Academy of Sciences, Beijing 100049, China\\
              \email{wangqiuyan@iie.ac.cn}\\
               D. Lin \at
             $2$. State Key Laboratory of Information Security, Institute of Information Engineering,
             Chinese Academy of Sciences,
              Beijing, 100195, China\\
              \email{ddlin@iie.ac.cn}   }            %  \\
%             \emph{Present address:} of F. Author  %  if needed
         %  \and
         %  S. Author \at
         %     second address
\date{Received: date / Accepted: date}
% The correct dates will be entered by the editor
\maketitle

\begin{abstract}
Recently, linear codes with few weights have been widely studied, since they have applications in data storage systems, communication systems and consumer electronics.  In this paper, we present a class of three-weight and five-weight linear codes over $\mathbb{F}_p$, where $p$ is an odd prime and  $\mathbb{F}_{p}$ denotes a finite field with $p$ elements.
The weight distributions of the linear codes constructed in this paper are also settled.  Moreover, the linear codes illustrated in the paper may have applications in  secret sharing schemes.

\end{abstract}

\section{Introduction and main results}
Let $q=p^{m}$ for an odd prime $p$ and a positive integer $m>2$.
Denote $ \mathbb{F}_q =\mathbb{F}_{p^{m}} $ the finite field with $ p^{m} $ elements and $\mathbb{F}_q^{*}=\mathbb{F}_{q}\backslash\{0\}$ the
multiplicative group of $\mathbb{F}_q$.

An $(n, M)$ code $\C$
over $\mathbb{F}_p$ is a subset of $\mathbb{F}_{p}^{n}$
 of size $M$. Among all kinds of codes, linear codes are  studied the most, since they are easier to describe,  encode and  decode than nonlinear codes.

A $[n,k,d]$ code $\C$ is called linear code over $\mathbb{F}_p$ if it is a $k$-dimensional subspace of
 $\mathbb{F}_{p}^{n}$ with minimum (Hamming) distance $d$. Usually, the vectors in $\C$ are called \textit{codewords}. The (Hamming) \textit{weight} $\wt(\mathbf{c})$
of a codeword $\mathbf{c}\in \C$
 is the number of nonzero coordinates in $\mathbf{c}$.
The \textit{weight enumerator} of $\C$ is a polynomial defined by
$$
1+A_1x+A_2x^{2}+\cdots+A_{n}x^{n},
$$
where $A_{i}$ denotes the number of codewords of weight $i$ in $\C$.
The \textit{weight distribution} $(A_0,A_1,\ldots, A_n)$    of $\C$ is of interest in coding theory and a lot of researchers are devoted to determining the weight distribution of specific codes. A code $\C$ is called a \textit{$t$-weight code} if $|\{i:A_i\neq 0,1\leq i\leq n\}|=t$. For the past decade years, a lot of codes with few weights are constructed \cite{Ding09,DY,DD14,DD15}.  Furthermore,  there is much literature on   the weight distribution of some special linear codes\cite{CKNC,Ding09,DLN,DY,LYL,LYF,ZDLZ,ZZD}.

Let $D=\{d_1,d_2,\ldots, d_n\} \subseteq \mathbb{F}_{q}$. A  linear code $\C_{D}$ of length $n$ over $\mathbb{F}_{p}$ is defined by
$$
\C_{D}=\{(\tr(xd_1), \tr(xd_2),\ldots, \tr(xd_{n})):x\in \mathbb{F}_{q}\},
$$
where $\tr$ denotes the absolute trace function over $\mathbb{F}_q$.
 The set $D$ is called the defining set of this code $ \C_{D}$. This construction was proposed by Ding et al. (see \cite{Ding15,DD14}) and  is
used  to obtain  linear codes with few weights \cite{DD15,WDX,X,ZLFH}.

In this paper, we set
\begin{eqnarray}\label{eq-a0}
        D=\{x\in \F_q^{*}: \tr(x^{2}+x)=0\}=\{d_1,d_2,\ldots, d_n\},  \nonumber \\
         \C_{D}=\{\mathbf{c}_{x}=(\tr(xd_1), \tr(xd_2), \ldots, \tr(xd_{n})): x\in \F_{q}\}
\end{eqnarray}
and determine the weight distribution of the proposed linear codes $\C_{D}$ of \eqref{eq-a0}.

 The parameters of the introduced linear codes $\C_{D}$ of \eqref{eq-a0} are described  in the following theorems. The proofs of the parameters will be presented later.

\begin{table}[ht]
\centering
\caption{The weight distribution of the codes of Theorem \ref{Thm1} }\label{tal:weightdistribution1}
\begin{tabular}{|l|l|}
\hline
% after \\: \hline or \cline{col1-col2} \cline{col3-col4} ...
\textrm{Weight} $w$ \qquad& \textrm{Multiplicity} $A$   \\
\hline
0 \qquad&   1  \\
\hline
$(p-1)p^{m-2}$ \qquad&  $p^{m-2}-1+p^{-1}(p-1)G$  \\
\hline
$(p-1)p^{m-2}+p^{-1}(p-1)G$  \qquad& $2(p-1)p^{m-2}-p^{-1}(p-1)G$\qquad\qquad  \\
\hline
$(p-1)p^{m-2}+p^{-1}(p-2)G$ \qquad&    $(p-1)^{2}p^{m-2}$ \\
\hline
\end{tabular}
\end{table}

\begin{Thm}\label{Thm1}
Let $m>2$ be even with $p\mid m. $ Then the code $\C_{D}$ of \eqref{eq-a0} is a $[p^{m-1}-1+p^{-1}(p-1)G, m]$ linear code with weight distribution in \autoref{tal:weightdistribution1},
where $G=-(-1)^{\frac{m(p-1)}{4}}p^{\frac{m}{2}}. $
\end{Thm}
\begin{Exp}
Let $(p,m)=(3,6)$. Then the corresponding code $\mathcal{C}_{D}$ has parameters $[ 260,6,162]$ and weight enumerator
$1+98x^{162}+324x^{171}+306x^{180}$.
\end{Exp}
\begin{Thm}\label{Thm2}
Let $m$ be even with $p\nmid m. $ Then the code $\C_{D}$ of \eqref{eq-a0} is a $[p^{m-1}-p^{-1}G-1, m]$ linear code with weight distribution in \autoref{tal:weightdistribution2},
where $G=-(-1)^{\frac{m(p-1)}{4}}p^{\frac{m}{2}}. $
\end{Thm}
\begin{Exp}
Let $(p,m)=(3,4)$. Then the corresponding code $\mathcal{C}_{D}$ has parameters $[ 29,4,18]$ and weight enumerator
$1+44x^{18}+30x^{21}+6x^{24}$. This code is optimal according to the codetables in \cite{MG}.
\end{Exp}

\begin{table}[h!]
\centering
\caption{The weight distribution of the codes of Theorem \ref{Thm2}.}\label{tal:weightdistribution2}
\begin{tabular}{|l|l|}
\hline
% after \\: \hline or \cline{col1-col2} \cline{col3-col4} ...
\textrm{Weight} $w$ \qquad& \textrm{Multiplicity} $A$   \\
\hline
0 \qquad&   1  \\
\hline
$(p-1)p^{m-2}-p^{-1}G$ \qquad\qquad&  $(p-1)(2p^{m-2}+p^{-1}G)$ \qquad \\
\hline
$(p-1)p^{m-2}$  \qquad& $\frac{1}{2}(p-1)(p^{m-1}-G)+p^{m-2}-1$ \qquad \\
\hline
$(p-1)p^{m-2}-2p^{-1}G$ \qquad\qquad &    $\frac{1}{2}(p^{2}-3p+2)(p^{m-2}+p^{-1}G)\qquad$\\
\hline
\end{tabular}
\end{table}

\begin{table}[ht]
\centering
\caption{The weight distribution of the codes of Theorem \ref{Thm3}.}\label{tal:weightdistribution3}
\begin{tabular}{|l|l|}
\hline
% after \\: \hline or \cline{col1-col2} \cline{col3-col4} ...
\textrm{Weight} $w$ \qquad& \textrm{Multiplicity} $A$   \\
\hline
0 \qquad&   1  \\
\hline
$(p-1)p^{m-2}$ \qquad&  $p^{m-1}-1$  \\
\hline
$(p-1)p^{m-2}+p^{\frac{m-3}{2}}$  \qquad& $\frac{1}{2}(p-1)^{2}p^{m-2}$ \\
\hline
$(p-1)p^{m-2}-p^{\frac{m-3}{2}}$ \qquad&  $\frac{1}{2}(p-1)^{2}p^{m-2}$  \\
\hline
$(p-1)p^{m-2}-(p-1)p^{\frac{m-3}{2}}$ \qquad\qquad&  $\frac{1}{2}(p-1)(p^{m-2}+p^{\frac{m-1}{2}})$ \qquad\qquad \\
\hline
$(p-1)p^{m-2}+(p-1)p^{\frac{m-3}{2}}$  \qquad& $\frac{1}{2}(p-1)(p^{m-2}-p^{\frac{m-1}{2}})$ \\
\hline
\end{tabular}
\end{table}

\begin{Thm}\label{Thm3}
If $m$ is odd and  $ p\mid m, $ then the linear code $\C_{D}$  of \eqref{eq-a0} has parameters $[p^{m-1}-1, m]$  and weight distribution in \autoref{tal:weightdistribution3}.
\end{Thm}
\begin{Exp}
Let $(p,m)=(3,3)$. Then the corresponding code $\mathcal{C}_{D}$ has parameters $[ 8,3,4]$ and weight enumerator
$1+6x^4+6x^5+8x^6+6x^7$. This code is almost optimal, since the optimal linear code has parameters $[8,3,5]$. By \autoref{tal:weightdistribution3}, $\mathcal{C}_{D}$ in Theorem \ref{Thm3} is a four weight linear code if and only if $p=m=3$.
\end{Exp}
\begin{Exp}
Let $(p,m)=(5,5)$. Then the corresponding code $\mathcal{C}_{D}$ has parameters  $[ 624,5,480]$ and weight enumerator
$1+300x^{480}+1000x^{495}+624x^{500}+1000x^{505}+200x^{520}$.
\end{Exp}
\begin{table}[ht]
\centering
\caption{The weight distribution of the codes of Theorem \ref{Thm4}.}\label{tal:weightdistribution4}
\begin{tabular}{|l|l|}
\hline
% after \\: \hline or \cline{col1-col2} \cline{col3-col4} ...
\textrm{Weight} $w$ \qquad& \textrm{Multiplicity} $A$   \\
\hline
0 \qquad&   1  \\
\hline
$(p-1)p^{m-2}+\frac{1}{p}(\frac{-m}{p})G\overline{G}$ \qquad&  $(p-1)(p^{m-2}-p^{-2}(\frac{-m}{p})G\overline{G})$  \\
\hline
$(p-1)p^{m-2}$  \qquad& $p^{m-2}+p^{-2}(\frac{-m}{p})(p-1)G\overline{G}-1$ \\
\hline
$(p-1)p^{m-2}+(\frac{-m}{p})p^{-2}(p-1)G\overline{G}$ \qquad&  $\frac{1}{2}(p-1)(p^{m-1}-(\frac{-m}{p})p^{-1}G\overline{G})$  \\
\hline
$(p-1)p^{m-2}+(\frac{-m}{p})p^{-2}(p+1)G\overline{G}$ \qquad\qquad &  $\frac{1}{2}(p-1)(p-2)(p^{m-2}-(\frac{-m}{p})p^{-2}G\overline{G})$\qquad   \\
\hline
$(p-1)p^{m-2}+p^{-2}(\frac{-m}{p})G\overline{G}$  \qquad& $(p-1)p^{m-2}+p^{-2}(\frac{-m}{p})(p-1)^{2}G\overline{G}$ \qquad \\
\hline
\end{tabular}
\end{table}

\begin{Thm}\label{Thm4}
If $m$ is odd and  $ p\nmid m, $ then the linear code $\C_{D}$ of \eqref{eq-a0} has parameters $[p^{m-1}+p^{-1}\left(\frac{-m}{p}\right)G\overline{G}-1, m]$  and weight distribution in \autoref{tal:weightdistribution4}, where $\left(\frac{\cdot}{\cdot}\right)$ is the Legendre symbol and $ G\overline{G}=(-1)^{\frac{(m+1)(p-1)}{4}}p^{\frac{m+1}{2}}. $
\end{Thm}
\begin{Exp}
Let $(p,m)=(3,5)$. Then the corresponding code $\mathcal{C}_{D}$ has parameters $[ 71,5,42]$ and weight enumerator
$1+30x^{42}+60x^{45}+90x^{48}+42x^{51}+20x^{54}$. We remark that this linear code is near optimal, since the corresponding optimal linear codes has parameters $[71,5,42]$.
\end{Exp}
{\bf Remark: }
In Theorem \ref{Thm4}, if $m=3$ and $ p \equiv 2 \mod3, $ the frequency of weight $ (p-1)p^{m-2} $
turns to be zero. Hence, in this case $ \C_{D} $ is a  four-weight linear code  with weight distribution in \autoref{tal:weightdistribution5}.
\begin{Exp}
Let $(p,m)=(5,3)$. Then the corresponding code $\mathcal{C}_{D}$ has parameters $[ 19,3,14]$ and weight enumerator
$1+36x^{14}+24x^{15}+60x^{16}+4x^{19}$. This code is optimal according to the datatables  in \cite{MG}.
\end{Exp}

\begin{table}[ht]
\centering
\caption{ The weight distribution of  $\C_{D}$, when   $m=3$ and $ p \equiv 2\pmod3$.}\label{tal:weightdistribution5}
\begin{tabular}{|l|l|}
\hline
% after \\: \hline or \cline{col1-col2} \cline{col3-col4} ...
\textrm{Weight} $w$\qquad \qquad\qquad& \textrm{Multiplicity} $A$ \qquad\qquad \qquad \\
\hline
0 \qquad\qquad\qquad\qquad&   1\qquad\qquad \qquad\qquad \\
\hline
$p^{2}-2p$\qquad\qquad \qquad&  $p^{2}-1$\qquad\qquad\qquad  \\
\hline
$p^{2}-2p+1$\qquad\qquad\qquad  \qquad& $\frac{1}{2}p(p^{2}-1)$\qquad\qquad \qquad\qquad\\
\hline
$p^{2}-2p-1$\qquad \qquad\qquad\qquad\qquad\qquad&  $\frac{1}{2}(p-2)(p^{2}-1)$\qquad\qquad \qquad\qquad\qquad\qquad \\
\hline
$p^{2}-p-1$\qquad\qquad \qquad\qquad&  $p-1$\qquad\qquad \qquad\qquad\qquad \\
\hline
\end{tabular}
\end{table}

\section{Preliminaries}
\label{}
In this section, we review some basic notations and results of group characters and
present some lemma which are needed for the proof of the main results.

An additive character $\chi$ of $\mathbb{F}_{q}$ is a mapping from $\mathbb{F}_{q}$ into the multiplicative group of complex numbers of absolute value $1$ with $\chi(g_1g_2)=\chi(g_1)\chi(g_2)$ for all  $g_1 , g_2\in \mathbb{F}_{q}$  \cite{LN}.

By Theorem 5.7 in \cite{LN}, for $b\in \mathbb{F}_{q}$,
\begin{equation}\label{eq-character}
\chi_{b}(x)=e^{\frac{2\pi \sqrt{-1}\tr(bx)}{p}}, \qquad\textrm{for\ all \ } x\in \mathbb{F}_{q}
\end{equation}
defines an additive character of $\mathbb{F}_{q}$, and all additive characters can be obtained in this way. Among the additive characters, we have the \textit{trivial character}    $\chi_0$ defined by $\chi_0(x)=1$ for all $x\in \mathbb{F}_{q}$; all other characters are called \textit{nontrivial}.  The character $\chi_1$ in \eqref{eq-character} will be called the \textit{canonical additive character}  of $\mathbb{F}_{q}$ \cite{LN}.

 The orthogonal property of additive characters   can be found  in  \cite{LN} and   is given as below
$$
\sum_{x\in \mathbb{F}_{q}}\chi(x)=\left\{\begin{array}{ll}
                                           q, & \textrm{if\ } \chi \textrm{\ is\ trivial},  \\
                                           0, & \textrm{if\ } \chi \textrm{\ is\ nontrivial}.
                                         \end{array}
                                         \right.
$$

Characters of the \textit{multiplicative group} $\mathbb{F}_{q}^{*}$ of $\mathbb{F}_{q}$  are called multiplicative character of $\mathbb{F}_{q}$. By Theorem 5.8 in \cite{LN}, for each $j=0,1,\ldots, q-2$, the function $\psi_{j}$ with
$$
\psi_{j}(g^{k})=e^{2\pi\sqrt{-1}jk/(q-1)} \textrm{for }\ k=0,1,\ldots, q-2
$$
defines a multiplicative character of $\mathbb{F}_{q}$, where $g$ is a generator of $\mathbb{F}_{q}^{*}$.
For $j=(q-1)/2$,  we have the \textit{quadratic character} $\eta= \psi_{(q-1)/2}$  defined by
$$
\eta(g^{k})=\left\{\begin{array}{ll}
                     -1, & \textrm{if } 2\nmid k,  \\
                     1,  & \textrm{if }   2\mid k.
                   \end{array}
                   \right.
$$
In the sequel,  we assume that $ \eta(0)=0 $.

We define the quadratic Gauss sum $G=G(\eta, \chi_1)$ over $\mathbb{F}_{q}$ by
$$G(\eta, \chi_1)=\sum_{x\in \mathbb{F}_{q}^{*}}\eta(x)\chi_1(x),$$
and the quadratic Gauss sum $\overline{G}=G(\overline{\eta}, \overline{\chi}_1)$ over $\mathbb{F}_{p}$ by
$$
G(\overline{\eta},\overline{\chi}_{1})=\sum_{x\in \mathbb{F}_{p}^{*}}\overline{\eta}(x)\overline{\chi}_1(x),
$$
where $\overline{\eta}$ and $\overline{\chi}_1$ denote the quadratic and canonical character of $\mathbb{F}_{p}, $ respectively.

The explicit values of quadratic Gauss sums are given as follows.
\begin{Lem}[\cite{LN}, Theorem 5.15]\label{Lem5}
Let the symbols be the same as before. Then
$$
G(\eta, \chi_1)=(-1)^{(m-1)}\sqrt{-1}^{\frac{(p-1)^{2}m}{4}}\sqrt{q}, \ \ G(\overline{\eta}, \overline{\chi}_1)=\sqrt{-1}^{\frac{(p-1)^{2}}{4}}\sqrt{p}.
$$
\end{Lem}

\begin{Lem} [\cite{DD14}, Lemma 7]\label{Lem6}
Let the symbols be the same as before. Then
\begin{enumerate}
\item if $m\geq2$ is even, then $\eta(y)=1$ for each $y\in \mathbb{F}_{p}^{*}$;
\item if $m$ is odd, then $\eta(y)=\overline{\eta}(y)$ for each $y\in \mathbb{F}_{p}^{*}$.
\end{enumerate}
\end{Lem}

\begin{Lem}[\cite{LN}, Theorem 5.33] \label{Lem7}
Let $\chi$ be a nontrivial additive character of $\mathbb{F}_{q}$, and let $f(x)=a_2x^{2}+a_1x+a_0\in \mathbb{F}_{q}[x]$ with $a_2\neq0$. Then
$$
\sum_{x\in \mathbb{F}_{q}}\chi(f(x))=\chi\left(a_0-a_1^{2}(4a_2)^{-1}\right)\eta(a_2)G(\eta,\chi).
$$
\end{Lem}

\begin{Lem}\label{Lem8}
Let the symbols be the same as before.  For $ y\in{F}_{p}^{*}, $  we have
$$
\sum_{y\in \mathbb{F}_{p}^{*}}\sum_{x\in \mathbb{F}_{q}}\zeta_{p}^{y\tr(x^{2}+x)}=\left\{\begin{array}{ll}
                                                                          (p-1)G, & \textrm{if\ } \ 2\mid m \ \textrm{and\ } p\mid m, \\
                                                                          -G, & \textrm{if\ } \ 2\mid m  \textrm{\ and\ } p\nmid m, \\
                                                                          0, & \textrm{if\ } \ 2\nmid m   \textrm{\ and\ } p\mid m, \\                                                                          \overline{\eta}(-m)G\overline{G}, & \textrm{if\ } \  2\mid m \textrm{\ and\ } p\nmid m.
                                                                        \end{array}
                                                                        \right.
$$
\end{Lem}

\begin{proof}
It follows from Lemma \ref{Lem7} that
\begin{align*}
\sum_{y\in \mathbb{F}_{p}^{*}}\sum_{x\in \mathbb{F}_{q}}\zeta_{p}^{y\tr(x^{2}+x)}&=\sum_{y\in \mathbb{F}_{p}^{*}}\sum_{x\in \mathbb{F}_{q}}\chi_1(yx^{2}+yx)\\
&=G\sum_{y\in \mathbb{F}_{p}^{*}}\chi_1\left(-\frac{y}{4}\right)\eta(y)\\
&=G\sum_{y\in \mathbb{F}_{p}^{*}}\eta(y)\zeta_{p}^{-\frac{y\tr(1)}{4}}.
\end{align*}
It is obviously that
$$
\tr(1)=m=\left\{\begin{array}{ll}
                  0,  & \textrm{if\ } p \mid m, \\
                  \neq0, & \textrm{otherwise}.
                \end{array}
                \right.
$$
 Consequently,
$$
\sum_{y\in \mathbb{F}_{p}^{*}}\sum_{x\in \mathbb{F}_{q}}\zeta_{p}^{y\tr(x^{2}+x)}=\left\{\begin{array}{ll}
                                               G\sum_{y\in \mathbb{F}_{p}^{*}}\eta(y),  & \textrm{if\ } 2\mid m\textrm{\ and\ } p\mid m,\\
                                               G\sum_{y\in \mathbb{F}_{p}^{*}}\zeta_{p}^{-\frac{ym}{4}},  & \textrm{if\ } 2\mid m\textrm{\ and\ } p\nmid m,\\
                                               G\sum_{y\in \mathbb{F}_{p}^{*}}\eta(y),  & \textrm{if\ } 2\nmid m\textrm{\ and\ } p\mid m,\\
                                               \eta(-m)G\sum_{y\in \mathbb{F}_{p}^{*}}\eta\left(-\frac{ym}{4}\right)\zeta_{p}^{-\frac{ym}{4}},  & \textrm{if\ } 2\nmid m\textrm{\ and\ } p\nmid m.
                                             \end{array}
                                             \right.
$$
Using Lemma \ref{Lem6}, we get this lemma.
\end{proof}

\begin{Lem}\label{Lem9}
Let the symbols be the same as before.  For $ b\in{F}_{q}^{*} $,  let
$$ B=\sum_{y\in \mathbb{F}_{p}^{*}}\sum_{z\in \mathbb{F}_{p}^{*}}\sum_{x\in \mathbb{F}_{q}}\chi_1(yx^{2}+yx+bzx). $$
Then
\begin{enumerate}

\item if $\tr(b^{2})\neq0$ and $\tr(b)=0$, we have $$B
=\left\{\begin{array}{ll}
                                                                           -(p-1)G, & \textrm{if\ } \ 2\mid m \textrm{\ and\ } p\mid m, \\
                                                                           \overline{\eta}\left(m\tr(b^{2})\right)G\overline{G}^{2}+G, & \textrm{if\ } \ 2\mid m \textrm{\ and\ } p\nmid m,  \\
                                                                           \overline{\eta}\left(-\tr(b^{2})\right)(p-1)G\overline{G}, & \textrm{if\ } \  2\nmid m \textrm{\ and\ } p\mid m, \\
                                                                           -\overline{\eta}\left(-\tr(b^{2})\right)+\overline{\eta}(-m))G\overline{G}, & \textrm{if\ } \ 2\nmid m \textrm{\ and\ }  p\nmid m;
                                                                         \end{array}
                                                                         \right.$$
\item if $\tr(b^{2})\neq0$ and $\tr(b)\neq0$, we have
 $$B
=\left\{\begin{array}{ll}
                                                                           \overline{\eta}(-1)G\overline{G}^{2}-(p-1)G, & \textrm{if\ }  2\mid m \textrm{\ and\ } p\mid m, \\
                                                                           G,  & \textrm{if\ }  2\mid m,\  p\nmid m \textrm{\ and\ } (\tr(b))^{2}=m\tr(b^{2}), \\
                                                                           \overline{\eta}(m\tr(b^{2})-(\tr(b))^{2})G\overline{G}^{2}+G, & \textrm{if\ }  2\mid m, \ p\nmid m\textrm{\ and\ } (\tr(b))^{2}\neq m\tr(b^{2}), \\
                                                                           -\overline{\eta}(-\tr(b^{2}))G\overline{G}, & \textrm{if\ }  2\nmid m\textrm{\ and\ }  p\mid m, \\
                                                                           (\overline{\eta}(-\tr(b^{2}))(p-1)-\overline{\eta}(-m))G\overline{G}, & \textrm{if\ }  2\nmid m, \  p\nmid m \textrm{\ and\ } (\tr(b))^{2}=m\tr(b^{2}), \\
                                                                           -(\overline{\eta}(-\tr(b^{2}))+\overline{\eta}(-m))G\overline{G}, & \textrm{if\ }  2\nmid m, \ p\nmid m \textrm{\ and\ } (\tr(b))^{2}\neq m\tr(b^{2});
                                                                         \end{array}
                                                                         \right.$$
\item if $\tr(b^{2})=0$ and $\tr(b)\neq0$, we have
$$
B=\left\{\begin{array}{ll}
                                                                           -(p-1)G, & \textrm{if\ } \ 2\mid m\textrm{\ and} \ p\mid m, \\
                                                                           G,  & \textrm{if\ } \ 2\mid m\textrm{\ and} \ p\nmid m, \\
                                                                           0, & \textrm{if\ } \ 2\nmid m\textrm{\ and} \ p\mid m, \\
                                                                           -\overline{\eta}(-m)G\overline{G}, & \textrm{if\ } \ 2\nmid m\textrm{\ and} \ p\nmid m;
                                                                         \end{array}
                                                                         \right.
                                                                         $$
\item if $\tr(b^{2})=0$ and $\tr(b)=0$, we have
$$B
=\left\{\begin{array}{ll}
                                                                           (p-1)^{2}G, & \textrm{if\ } \ 2\mid m\textrm{\ and} \ p\mid m, \\
                                                                           -(p-1)G,  & \textrm{if\ } \ 2\mid m\textrm{\ and} \ p\nmid m, \\
                                                                           0, & \textrm{if\ } \ 2\nmid m\textrm{\ and} \ p\mid m, \\
                                                                           \overline{\eta}(-m)(p-1)G\overline{G}, & \textrm{if\ } \ 2\nmid m\textrm{\ and} \ p\nmid m.
                                                                         \end{array}
                                                                         \right.$$
\end{enumerate}
\end{Lem}
\begin{proof}
We only give the proof of the first part since the remaining parts are similar.

By Lemma \ref{Lem7}, we have
\begin{align*}
B&=G\sum_{y\in \mathbb{F}_{p}^{*}}\sum_{z\in \mathbb{F}_{p}^{*}}\eta(y)\chi_1\left(-\frac{(y+bz)^{2}}{4y}\right)\\
&=G\sum_{y\in \mathbb{F}_{p}^{*}}\eta(y)\chi_1\left(-\frac{y}{4}\right)\sum_{z\in \mathbb{F}_{p}^{*}}\chi_1\left(-\frac{b^{2}z^{2}}{4y}-\frac{bz}{2}\right)\\
&=G\sum_{y\in \mathbb{F}_{p}^{*}}\eta(y)\chi_1\left(-\frac{y}{4}\right)\sum_{z\in \mathbb{F}_{p}^{*}}\zeta_{p}^{-\frac{\tr(b^{2})z^{2}}{4y}-\frac{\tr(b)z}{2}}.
\end{align*}
Note that in the first part, $\tr(b^{2})\neq0$ and $\tr(b)=0.$ Therefore,
\begin{align*}
B&=G\sum_{y\in \mathbb{F}_{p}^{*}}\eta(y)\chi_1\left(-\frac{y}{4}\right)\sum_{z\in \mathbb{F}_{p}^{*}}\overline{\chi}_1\left(-\frac{\tr(b^{2})z^{2}}{4y}\right)&\\
&=G\sum_{y\in \mathbb{F}_{p}^{*}}\eta(y)\chi_1\left(-\frac{y}{4}\right)\sum_{z\in \mathbb{F}_{p}}\overline{\chi}_1\left(-\frac{\tr(b^{2})z^{2}}{4y}\right)-G\sum_{y\in \mathbb{F}_{p}^{*}}\eta(y)\chi_1\left(-\frac{y}{4}\right)&\\
&=G\sum_{y\in \mathbb{F}_{p}^{*}}\eta(y)\chi_1\left(-\frac{y}{4}\right)\overline{\chi}_1(0)\overline{\eta}(-\tr(b^{2})y)\overline{G}-G\sum_{y\in \mathbb{F}_{p}^{*}}\eta(y)\chi_1\left(-\frac{y}{4}\right)&\\
&=\overline{\eta}(-\tr(b^{2}))G\overline{G}\sum_{y\in \mathbb{F}_{p}^{*}}\eta(y)\chi_1\left(-\frac{y}{4}\right)\overline{\eta}(y)-G\sum_{y\in \mathbb{F}_{p}^{*}}\eta(y)\chi_1\left(-\frac{y}{4}\right)&\\
&=\left\{\begin{array}{ll}
 \overline{\eta}(-\tr(b^{2}))G\overline{G}\sum_{y\in \mathbb{F}_{p}^{*}}\overline{\eta}(y)-G\sum_{y\in \mathbb{F}_{p}^{*}}1, & \textrm{if\ } \ 2\mid m\textrm{\ and} \ p\mid m, \\
                                                                           \overline{\eta}(m\tr(b^{2}))G\overline{G}\sum_{y\in \mathbb{F}_{p}^{*}}\overline{\chi_1}\left(-\frac{my}{4}\right)\overline{\eta}\left(-\frac{my}{4}\right)-G\sum_{y\in \mathbb{F}_{p}^{*}}\zeta_{p}^{-\frac{my}{4}}, & \textrm{if\ } \ 2\mid m\textrm{\ and} \ p\nmid m,  \\
                                                                           \overline{\eta}(-\tr(b^{2}))G\overline{G}\sum_{y\in \mathbb{F}_{p}^{*}}1-G\sum_{y\in \mathbb{F}_{p}^{*}}\overline{\eta}(y), & \textrm{if\ } \ 2\nmid m\textrm{\ and} \ p\mid m, \\
                                                                           \overline{\eta}(-\tr(b^{2}))G\overline{G}\sum_{y\in \mathbb{F}_{p}^{*}}\zeta_{p}^{-\frac{my}{4}}-G\overline{\eta}(-m)\sum_{y\in \mathbb{F}_{p}^{*}}\overline{\eta}\left(-\frac{my}{4}\right)\zeta_{p}^{-\frac{my}{4}}, & \textrm{if\ } \ 2\nmid m\textrm{\ and} \ p\nmid m.
                                                                         \end{array}
                                                                         \right.
\end{align*}
 Combining Lemma \ref{Lem6} and the equation $ \sum_{y\in \mathbb{F}_{p}^{*}}\zeta_{p}^{y}=-1, $ we get the result of the first part.                                                                          \end{proof}

\begin{Lem}\label{Lem10}
For $ a\in \mathbb{F}_{p} $, let
$$ N(0,a)=\{x\in {F}_{q}: \tr(x^{2})=0, \tr(x)=a\}. $$
Then
\begin{enumerate}
\item if $a\neq0,$ we have
$$|N(0,a)|=\left\{\begin{array}{ll}
                                                                          p^{m-2}, & \textrm{if\ } \ p\mid m, \\
                                                                          p^{m-2}+p^{-1}G, & \textrm{if\ } \ 2\mid m\textrm{\ and} \ p\nmid m, \\
                                                                          p^{m-2}-p^{-2}\overline{\eta}(-m)G\overline{G}, & \textrm{if\ } \ 2\nmid m\textrm{\ and}\  p\nmid m;
                                                                        \end{array}
                                                                        \right.$$

\item if $a=0,$ we have
$$|N(0,0)|=\left\{\begin{array}{ll}
                                                                          p^{m-2}+p^{-1}(p-1)G, & \textrm{if\ } \ 2\mid m\textrm{\ and} \ p\mid m, \\
                                                                          p^{m-2}, & \textrm{if\ } \ 2\mid m\textrm{\ and} \ p\nmid m, \\
                                                                          p^{m-2}, & \textrm{if\ } \ 2\nmid m\textrm{\ and} \ p\mid m, \\                                                                          p^{m-2}+p^{-2}\overline{\eta}(-m)(p-1)G\overline{G}, & \textrm{if\ } \ 2\nmid m\textrm{\ and} \ p\nmid m.
                                                                        \end{array}
                                                                        \right.$$
\end{enumerate}
\end{Lem}

\begin{proof}
We only prove the first statement of this lemma, since the other statements can be similarly proved.

 For $a\in \mathbb{F}_{p}^{*},$ we have
\begin{align*}
|N(0,a)|  &= p^{-2}\sum_{x\in \mathbb{F}_{q}}\left(\sum_{y \in F_{p}}\zeta_{p}^{y\tr(x^{2})}\right)
\left(\sum_{z \in F_{p}}\zeta_{p}^{z(\tr(x)-a)}\right)   \\
  &=p^{-2}\sum_{x\in \mathbb{F}_{q}}\left(1+\sum_{y \in F_{p}^{\ast}}\zeta_{p}^{y\tr(x^{2})}\right)
\left(1+\sum_{z \in F_{p}^{\ast}}\zeta_{p}^{z(\tr(x)-a)}\right) \\
 &= p^{m-2} + p^{-2}\sum_{y \in F_{p}^{\ast}}\sum_{x\in \mathbb{F}_{q}}\zeta_{p}^{y\tr(x^{2})}
 + p^{-2}\sum_{z \in F_{p}^{\ast}}\sum_{x\in \mathbb{F}_{q}}\zeta_{p}^{z(\tr(x)-a)}  \\
& \quad +
p^{-2}\sum_{y \in F_{p}^{\ast}}\sum_{z \in F_{p}^{\ast}}\sum_{x\in \mathbb{F}_{q}}\zeta_{p}^{\tr(yx^{2}+zx)-za} \\
& = p^{m-2} + p^{-2}\sum_{y \in F_{p}^{\ast}}\sum_{x\in \mathbb{F}_{q}}\chi_1(yx^{2}) +
p^{-2}\sum_{y \in F_{p}^{\ast}}\sum_{z \in F_{p}^{\ast}}\zeta_{p}^{-za}\sum_{x\in \mathbb{F}_{q}}\chi_1(yx^{2}+zx).
\end{align*}

By Lemma \ref{Lem7}, we obtain
\begin{align*}
 |N(0,a)|&= p^{m-2} + p^{-2}\sum_{y \in F_{p}^{\ast}}\sum_{x\in \mathbb{F}_{q}}\chi_1(yx^{2}) +
p^{-2}\sum_{y \in F_{p}^{\ast}}\sum_{z \in F_{p}^{\ast}}\zeta_{p}^{-za}\sum_{x\in \mathbb{F}_{q}}\chi_1(yx^{2}+zx)& \\
 &= p^{m-2} + p^{-2}\sum_{y \in F_{p}^{\ast}}\chi_1(0)\eta(y)G +
p^{-2}\sum_{y \in F_{p}^{\ast}}\sum_{z \in F_{p}^{\ast}}\zeta_{p}^{-za}\chi_1(-\frac{z^{2}}{4y})\eta(y)G &\\
&=\left\{\begin{array}{ll}
                                                                          p^{m-2} + p^{-2}G\sum_{y \in F_{p}^{\ast}}\eta(y) +
p^{-2}G\sum_{y \in F_{p}^{\ast}}\sum_{z \in F_{p}^{\ast}}\zeta_{p}^{-za}\eta(y), & \textrm{if\ } \ p\mid m \\
                                                                          p^{m-2} + p^{-2}G\sum_{y \in F_{p}^{\ast}}\eta(y) +
p^{-2}G\sum_{y \in F_{p}^{\ast}}\sum_{z \in F_{p}^{\ast}}\zeta_{p}^{-za}\eta(y)\zeta_{p}^{-\frac{mz^{2}}{4y}}, & \textrm{if\ } \ p\nmid m
\end{array}
\right.\\
&=\left\{\begin{array}{ll}
                                                                          p^{m-2}, & \textrm{if\ } \ p\mid m \\
                                                                          p^{m-2} + p^{-2}(p-1)G +
p^{-2}G\sum_{y \in F_{p}^{\ast}}\sum_{z \in F_{p}^{\ast}}\zeta_{p}^{-za}\zeta_{p}^{-\frac{mz^{2}}{4y}}, & \textrm{if\ } \ 2\mid m \textrm{\ and\ } p\nmid m \\
p^{m-2} +
p^{-2}\eta(-m)G\sum_{y \in F_{p}^{\ast}}\sum_{z \in F_{p}^{\ast}}\zeta_{p}^{-za}\eta\left(-\frac{mz^{2}}{4y}\right)\zeta_{p}^{-\frac{mz^{2}}{4y}}, & \textrm{if\ } \ 2\nmid m \textrm{\ and\ } p\nmid m
                                                                        \end{array}
                                                                        \right.\\
&=\left\{\begin{array}{ll}
                                                                          p^{m-2}, & \textrm{if\ } \ p\mid m, \\
                                                                          p^{m-2} + p^{-1}G, & \textrm{if\ } \ 2\mid m \textrm{\  and\ } p\nmid m,\\
p^{m-2}
-p^{-2}\overline{\eta}(-m)G\overline{G}, & \textrm{if\ } \ 2\nmid m \textrm{\  and\  } p\nmid m. \\
\end{array}                                                                       \right.
\end{align*}
\end{proof}

\begin{Lem} \label{Lem11}
 Let
$$ N(0,\overline{0})=\{x\in {F}_{q}: \tr(x^{2})=0\textrm{\ and\ } \tr(x)\neq0\},$$
$$N(\overline{0},\overline{0})=\{x\in {F}_{q}: \tr(x^{2})\neq0\textrm{\ and\ } \tr(x)\neq0\},$$
$$ N(\overline{0},0)=\{x\in {F}_{q}: \tr(x^{2})\neq0\textrm{\ and\ } \tr(x)=0\}. $$
Then we get
\begin{enumerate}
\item
$$
|N(0,\overline{0})|=\left\{\begin{array}{ll}
                                                                          (p-1)p^{m-2}, & \textrm{if\ } p\mid m, \\
                                                                          (p-1)\left(p^{m-2}+p^{-1}G\right), & \textrm{if\ }  2\mid m\textrm{\ and\ }  p\nmid m, \\
                                                                          (p-1)\left(p^{m-2}-p^{-2}\overline{\eta}(-m)G\overline{G}\right), & \textrm{if\ }  2\nmid m\textrm{\ and\ } p\nmid m;
                                                                        \end{array}
                                                                        \right.
$$
\item
$$
|N(\overline{0},\overline{0})|=\left\{\begin{array}{ll}
                                                                          (p-1)^{2}p^{m-2}, & \textrm{if\ }  p\mid m, \\
                                                                          (p-1)^{2}p^{m-2}-(p-1)p^{-1}G, & \textrm{if\ }  2\mid m\textrm{\ and\ }  p\nmid m, \\
                                                                          (p-1)^{2}p^{m-2}+(p-1)p^{-2}\overline{\eta}(-m)G\overline{G}, & \textrm{if\ }  2\nmid m\textrm{\ and\ }  p\nmid m;
                                                                        \end{array}
                                                                        \right.
$$
\item
$$
|N(\overline{0},0)|=\left\{\begin{array}{ll}
                                                                          (p-1)p^{m-2}-p^{-1}(p-1)G, & \textrm{if\ } 2\mid m\textrm{\ and\ }  p\mid m, \\
                                                                          (p-1)p^{m-2}, & \textrm{if\ }  2\mid m\textrm{\ and\ } p\nmid m, \\
                                                                          (p-1)p^{m-2}, & \textrm{if\ } 2\nmid m\textrm{\ and\ } p\mid m, \\                                                                          (p-1)p^{m-2}-p^{-2}\overline{\eta}(-m)(p-1)G\overline{G}, & \textrm{if\ }  2\nmid m\textrm{\ and\ }  p\nmid m.
                                                                        \end{array}
                                                                        \right.
$$
\end{enumerate}
\end{Lem}

\begin{proof}
By the definitions,  we have
$$|N(\overline{0},0)|+|N(0,0)|=p^{m-1},$$
$$|N(0,\overline{0})|=\sum_{a \in F_{p}^{\ast}}|N(0,a)|,$$
$$ |N(\overline{0},\overline{0})|+|N(0,\overline{0})|=p^{m}-p^{m-1}.$$
Then the desired results follow from Lemma \ref{Lem10}.
\end{proof}

\begin{Lem}\label{Lem12}
\ Suppose  $ p\nmid m $ and let  $$ V=\{x\in {F}_{q}: \tr(x)\neq0\textrm{\ and\ } (\tr(x))^{2}= m\tr(x^{2})\}. $$
Then
$$
|V|=\left\{\begin{array}{ll}
                                                                          (p-1)p^{m-2}, & \textrm{if\ } \ 2\mid m, \\
                                                                          (p-1)p^{m-2}+p^{-2}\overline{\eta}(-m)(p-1)^{2}G\overline{G}, & \textrm{if\ } \ 2\nmid m.
                                                                        \end{array}
                                                                        \right.
$$
\end{Lem}
\begin{proof}
For $c\in \mathbb{F}_{p}^{*},$ set
$$ S_{c}=\{x\in \mathbb{F}_{p}:\tr(x)=c\textrm{\ and\ }\tr(x^{2})=c^{2}/m\}. $$
Then
$$ |V|=\sum_{c \in F_{p}^{\ast}}|S_{c}|. $$
By definition,  we have
\begin{align*}
|S_{c}|  &= p^{-2}\sum_{x\in \mathbb{F}_{q}}\left(\sum_{y \in F_{p}}\zeta_{p}^{y(\tr(x^{2})-\frac{c^{2}}{m})}\right)
\left(\sum_{z \in F_{p}}\zeta_{p}^{z(\tr(x)-c)}\right)   \\
  &=p^{-2}\sum_{x\in \mathbb{F}_{q}}\left(1+\sum_{y \in F_{p}^{\ast}}\zeta_{p}^{y(\tr(x^{2})-\frac{c^{2}}{m})}\right)
\left(1+\sum_{z \in F_{p}^{\ast}}\zeta_{p}^{z(\tr(x)-c)}\right) \\
 &= p^{m-2} + p^{-2}\sum_{y \in F_{p}^{\ast}}\sum_{x\in \mathbb{F}_{q}}\zeta_{p}^{y(\tr(x^{2})-\frac{c^{2}}{m})} +
p^{-2}\sum_{y \in F_{p}^{\ast}}\sum_{z \in F_{p}^{\ast}}\sum_{x\in \mathbb{F}_{q}}\zeta_{p}^{\tr(yx^{2}+zx)-\frac{yc^{2}}{m}-zc}.
\end{align*}
Let
$$ s_{c}=\sum_{y \in F_{p}^{\ast}}\sum_{x\in \mathbb{F}_{q}}\zeta_{p}^{y\left(\tr(x^{2})-\frac{c^{2}}{m}\right)},$$
$$ \overline{s}_{c}=\sum_{y \in F_{p}^{\ast}}\sum_{z \in F_{p}^{\ast}}\sum_{x\in \mathbb{F}_{q}}\zeta_{p}^{\tr(yx^{2}+zx)-\frac{yc^{2}}{m}-zc}. $$
It is straightforward to have that
$$ s_{c}=\sum_{y \in F_{p}^{\ast}}\zeta_{p}^{-\frac{yc^{2}}{m}}\sum_{x\in \mathbb{F}_{q}}\chi_{1}(yx^{2}). $$
By Lemma \ref{Lem7}, we obtain
\begin{align}\label{eq-s1}
 s_{c}&=\sum_{y \in F_{p}^{\ast}}\zeta_{p}^{-\frac{yc^{2}}{m}}\chi_1(0)\eta(y)G & \nonumber\\
&=\left\{\begin{array}{ll}
                                                                          G\sum_{y \in F_{p}^{\ast}}\zeta_{p}^{-\frac{yc^{2}}{m}}, & \textrm{if\ }  2\mid m \\
                                                                          \eta(-m)G\sum_{y \in F_{p}^{\ast}}\eta\left(-\frac{yc^{2}}{m}\right)\zeta_{p}^{-\frac{yc^{2}}{m}}, & \textrm{if\ }  2\nmid m
                                                                        \end{array}
                                                                        \right.\nonumber
                                                                        \\
&=\left\{\begin{array}{ll}
                                                                          -G, & \textrm{if\ }  2\mid m, \\
                                                                          \overline{\eta}(-m)G\overline{G}, & \textrm{if\ }  2\nmid m.
                                                                        \end{array}
                                                                        \right.
\end{align}
Meanwhile,
\begin{align*}
 \overline{s}_{c}&=\sum_{y \in F_{p}^{\ast}}\sum_{z \in F_{p}^{\ast}}\zeta_{p}^{-\frac{yc^{2}}{m}-zc}\sum_{x\in \mathbb{F}_{q}}\chi_{1}(yx^{2}+zx) \\
 &=\sum_{y \in F_{p}^{\ast}}\sum_{z \in F_{p}^{\ast}}\zeta_{p}^{-\frac{yc^{2}}{m}-zc}\chi_1\left(-\frac{z^{2}}{4y}\right)\eta(y)G.
\end{align*}
Hence,
\begin{align}\label{eq-s3}
 \sum_{c \in F_{p}^{\ast}}\overline{s}_{c}&=G\sum_{y \in F_{p}^{\ast}}\sum_{z \in F_{p}^{\ast}}\chi_1\left(-\frac{z^{2}}{4y}\right)\eta(y)\sum_{c \in F_{p}^{\ast}}\zeta_{p}^{-\frac{yc^{2}}{m}-zc}& \nonumber\\
 &=G\sum_{y \in F_{p}^{\ast}}\sum_{z \in F_{p}^{\ast}}\chi_1\left(-\frac{z^{2}}{4y}\right)\eta(y)\sum_{c \in F_{p}}\overline{\chi}_1\left(-\frac{yc^{2}}{m}-zc\right)-G\sum_{y \in F_{p}^{\ast}}\sum_{z \in F_{p}^{\ast}}\chi_1\left(-\frac{z^{2}}{4y}\right)\eta(y)& \nonumber\\
& =G\sum_{y \in F_{p}^{\ast}}\sum_{z \in F_{p}^{\ast}}\chi_1\left(-\frac{z^{2}}{4y}\right)\eta(y)\overline{\chi}_{1}\left(\frac{mz^{2}}{4y}\right)\overline{\eta}(-my)\overline{G}-G\sum_{y \in F_{p}^{\ast}}\sum_{z \in F_{p}^{\ast}}\chi_1\left(-\frac{z^{2}}{4y}\right)\eta(y) & \nonumber\\
&=\overline{\eta}(-m)G\overline{G}\sum_{y \in F_{p}^{\ast}}\sum_{z \in F_{p}^{\ast}}\eta(y)\overline{\eta}(y)-G\sum_{y \in F_{p}^{\ast}}\sum_{z \in F_{p}^{\ast}}\zeta_{p}^{-\frac{mz^{2}}{4y}}\eta(y) &\nonumber\\
& =\left\{\begin{array}{ll}
                                                                          (p-1)G , & \textrm{if\ }  2\mid m, \\
                                                                          \overline{\eta}(-m)(p-1)^{2}G\overline{G}-\overline{\eta}(-m)(p-1)G\overline{G}, & \textrm{if\ }  2\nmid m.
                                                                        \end{array}
                                                                        \right.
\end{align}
 We get
 \begin{align*}
 |V|&=\sum_{c \in F_{p}^{\ast}}|S_{c}|=\sum_{c \in F_{p}^{\ast}}(p^{m-2}+p^{-2}s_{c}+p^{-2}\overline{s}_{c})\\
 &=\sum_{c \in F_{p}^{\ast}}p^{m-2}+p^{-2}\sum_{c \in F_{p}^{\ast}}s_{c}+p^{-2}\sum_{c \in F_{p}^{\ast}}\overline{s}_{c}.
  \end{align*}
  By \eqref{eq-s1} and \eqref{eq-s3},  we get this lemma.
\end{proof}

\section{Proof of main results}
In this section, we will present a class of linear codes with three weights and  five weights over $\mathbb{F}_{p}$.

Recall that the defining set considered in this paper is defined  by
 $$
 D=\{x\in \mathbb{F}_{q}^{*}: \tr(x^{2}+x)=0\}.
 $$
Let $ n_{0}=|D|+1. $ Then
\begin{align*}
n_{0} &=\frac{1}{p}\sum_{x \in F_{q}}\left(\sum_{y \in F_{p}}\zeta_{p}^{y\tr(x^{2}+x)}\right)& \\
& = p^{m-1}+\frac{1}{p}\sum_{x \in F_{q}}\sum_{y \in F_{p}^{\ast}}\zeta_{p}^{y\tr(x^{2}+x)}. &
\end{align*}

Define $N_{b}=\{x\in \mathbb{F}_{q}: \tr(x^{2}+x)=0\textrm{\  and\ } \tr(bx)=0\}.$
Let $ \wt(\mathbf{c}_b)$ denote the Hamming weight of the  codeword
 $
 \mathbf{c}_b
 $
 of the code $\C_{D}.$ It can be easily checked that
\begin{equation} \label{eq-3.1}
\wt(\mathbf{c}_b)=n_0-|N_{b}|.
\end{equation}
For $b\in \mathbb{F}_{q}^{*},$ we have
\begin{align}\label{eq-weight}
|N_{b}| &= p^{-2}\sum_{x\in \mathbb{F}_{q}}\left(\sum_{y \in F_{p}}\zeta_{p}^{y\tr(x^{2}+x)}\right)
\left(\sum_{z \in F_{p}}\zeta_{p}^{z\tr(bx)}\right)  \nonumber \\
 &=p^{-2}\sum_{x\in \mathbb{F}_{q}}\left(1+\sum_{y \in F_{p}^{\ast}}\zeta_{p}^{y\tr(x^{2}+x)}\right)
\left(1+\sum_{z \in F_{p}^{\ast}}\zeta_{p}^{z\tr(bx)}\right)\nonumber \\
&= p^{m-2} + p^{-2}\sum_{y \in F_{p}^{\ast}}\sum_{x\in \mathbb{F}_{q}}\zeta_{p}^{y\tr(x^{2}+x)}
+ p^{-2}\sum_{z \in F_{p}^{\ast}}\sum_{x\in \mathbb{F}_{q}}\zeta_{p}^{z\tr(bx)} \nonumber \\
&\quad +
p^{-2}\sum_{y \in F_{p}^{\ast}}\sum_{z \in F_{p}^{\ast}}\sum_{x\in \mathbb{F}_{q}}\zeta_{p}^{\tr(yx^{2}+yx+bzx)} \nonumber\\
&= p^{m-2} + p^{-2}\sum_{y \in F_{p}^{\ast}}\sum_{x\in \mathbb{F}_{q}}\zeta_{p}^{y\tr(x^{2}+x)} +
p^{-2}\sum_{y \in F_{p}^{\ast}}\sum_{z \in F_{p}^{\ast}}\sum_{x\in \mathbb{F}_{q}}\zeta_{p}^{\tr(yx^{2}+yx+bzx)}.
\end{align}

Our task in this section is to  calculate $n_{0}$, $|N_{b}|$ and give the proof of the main results.

\subsection{The first case of three-weight linear codes}
In this subsection,  suppose $ 2\mid m$ and $p\mid m. $
To determine the weight distribution of $\C_{D}$ of \eqref{eq-a0}, the following  lemma is needed.

\begin{Lem}\label{Lem13}
 Let $ b\in \mathbb{F}_{q}^{*}. $ Then
$$
|N_{b}|=\left\{\begin{array}{ll}
                                                                          p^{m-2},& \ \textrm{if\ }  \tr(b^{2})=0\textrm{\ and\ }\tr(b)\neq0 \\
                                                                          & \ \textrm{or\ }\tr(b^{2})\neq0\textrm{\ and\ }\tr(b)=0, \\
                                                                          p^{m-2}-(-1)^{\frac{m(p-1)}{4}}(p-1)p^{\frac{m-2}{2}}, & \textrm{if\ } \tr(b^{2})=0\textrm{\ and\ }\tr(b)=0, \\
                                                                          p^{m-2}-(-1)^{\frac{m(p-1)}{4}}p^{\frac{m-2}{2}}, & \textrm{if\ }  \tr(b^{2})\neq0\textrm{\ and\ }\tr(b)\neq0 .
                                                                        \end{array}
                                                                        \right.
$$
\end{Lem}

\begin{proof}
The desired result follows directly from \eqref{eq-weight}, Lemmas \ref{Lem5},  \ref{Lem8} and  \ref{Lem9}. We omit the details.
\end{proof}

After the preparations above, we proceed to prove Theorem \ref{Thm1}.
 By Lemma \ref{Lem8}, if $ 2\mid m$ and  $p\mid m$,  we have
$$
n_{0}=p^{m-1}+p^{-1}(p-1)G.
$$
  Combining  \eqref{eq-3.1}, \eqref{eq-weight} and Lemma \ref{Lem13}, we get
\begin{align*}
 &\wt(c_b)=n_0-|N_b| \\
 &\in\left\{(p-1)p^{m-2}+p^{-1}(p-1)G, (p-1)p^{m-2}, (p-1)p^{m-2}+p^{-1}(p-2)G\right\}.
\end{align*}
Set
\begin{align*}
 \omega_{1}&=(p-1)p^{m-2}+p^{-1}(p-1)G,\\
  \omega_{2}&=(p-1)p^{m-2}, \\
  \omega_{3}&=(p-1)p^{m-2}+p^{-1}(p-2)G.
 \end{align*}
By Lemma \ref{Lem13}, we obtain
\begin{align*}
 A_{\omega_{1}}&=|N(0,\overline{0})|+|N(\overline{0},0)|,\\
 A_{\omega_{2}}&=|N(0,0)|-1,  \\
 A_{\omega_{3}}&=|N(\overline{0},\overline{0})|.
 \end{align*}
Then the results in Theorem \ref{Thm1} follow from Lemmas \ref{Lem5} and  \ref{Lem11}.

\subsection{The second case of three-weight linear codes}
In this subsection,  assume   $ 2\mid m$ and $p\nmid m$. By \eqref{eq-weight}, Lemmas \ref{Lem8} and  \ref{Lem9}, it is easy to get the following lemma.

\begin{Lem}\label{Lem14}
\ Let $ b\in \mathbb{F}_{q}^{*}$ and the symbols be the same as before.  Then we have
$$
|N_{b}|=\left\{\begin{array}{ll}
                                                                          p^{m-2}, & \textrm{if\ }  \tr(b^{2})=0\textrm{\ and\ }\tr(b)\neq0 \textrm{\ or\ }\\
                                                                          & \tr(b^{2})\neq0 \textrm{\ and\ } (\tr(b))^{2}=m\tr(b^{2}), \\
                                                                          p^{m-2}-p^{-1}G, & \textrm{if\ }  \tr(b^{2})=0 \textrm{\ and\ }\tr(b)=0,  \\
                                                                          p^{m-2}+p^{-2}\overline{\eta}(m\tr(b^{2})-(\tr(b))^{2})G\overline{G}^{2}, &\textrm{if\ } \tr(b^{2})\neq0\textrm{\ and\ } (\tr(b))^{2}\neq m\tr(b^{2}), \\
                                                                          p^{m-2}+p^{-2}\overline{\eta}(m\tr(b^{2}))G\overline{G}^{2}, &\textrm{if\ } \tr(b^{2})\neq0\textrm{\ and\ }\tr(b)=0.
                                                                        \end{array}
                                                                        \right.
$$
\end{Lem}

We are now turning to the proof of  Theorem \ref{Thm2}.
If  $ 2\mid m$ and  $p\nmid m $,  by Lemma \ref{Lem8}, we have
$$ n_{0}=p^{m-1}-p^{-1}G. $$

It follows from   \eqref{eq-3.1} and Lemma \ref{Lem14} that
$$ wt(c_b)\in \left\{(p-1)p^{m-2}-p^{-1}G, (p-1)p^{m-2}, (p-1)p^{m-2}-2p^{-1}G\right\}.$$
Suppose
\begin{align*}
\omega_{1}&=(p-1)p^{m-2}-p^{-1}G, \\
\omega_{2}&=(p-1)p^{m-2}, \\
\omega_{3}&=(p-1)p^{m-2}-2p^{-1}G.
\end{align*}
By Lemmas \ref{Lem11},  \ref{Lem12} and  \ref{Lem14}, we have
$$ A_{\omega_{1}}=|N(0,\overline{0})|+|V|=(p-1)(2p^{m-2}+p^{-1}G). $$
It is easy to check  that the minimum distance of the dual code $ \mathcal{C}^{\perp}_{D} $ of $\mathcal{C}_{D}$ is equal to $ 2 $.  By the first two Pless Power Moments(\cite{HP}, p. 260) the frequency $A_{w_i}$ of $w_{i}$ satisfies the following equations:
\begin{eqnarray}\label{eq-3.2}
\left\{\begin{array}{l}
         A_{w_1}+A_{w_2}+A_{w_3}=p^{m}-1, \\
         w_1A_{w_1}+w_2A_{w_2}+w_3A_{w_3}=p^{m-1}(p-1)n,
       \end{array}
       \right.
\end{eqnarray}
where $n=p^{m-1}-p^{-1}G-1$.
 A simple calculation leads to  the weight distribution of \autoref{tal:weightdistribution2}.
The proof of Theorem \ref{Thm2} is completed.

\subsection{The first case of 5-weight linear codes}
In this subsection,  set  $ 2\nmid m$ and $p\mid m $.  By \eqref{eq-weight}, Lemmas \ref{Lem8} and  \ref{Lem9}, we get the following lemma.
\begin{Lem} \label{Lem15}
\ Let $ b\in \mathbb{F}_{q}^{*}, $ then
$$
|N_{b}|=\left\{\begin{array}{ll}
                                                                          p^{m-2}, & \textrm{if\ } \ \tr(b^{2})=0, \\
                                                                          p^{m-2}-p^{-2}\overline{\eta}(-1)G\overline{G}, & \textrm{if\ } \ \tr(b^{2})\neq0,\ \tr(b)\neq0, \   \overline{\eta}(\tr(b^{2}))=1,\\
                                                                          p^{m-2}+p^{-2}\overline{\eta}(-1)G\overline{G}, & \textrm{if\ } \ \tr(b^{2})\neq0,\ \tr(b)\neq0,  \  \overline{\eta}(\tr(b^{2}))=-1,\\
                                                                          p^{m-2}+p^{-2}\overline{\eta}(-1)(p-1)G\overline{G}, & \textrm{if\ } \ \tr(b^{2})\neq0,\ \tr(b)=0,\    \overline{\eta}(\tr(b^{2}))=1,\\
                                                                          p^{m-2}-p^{-2}\overline{\eta}(-1)(p-1)G\overline{G}, & \textrm{if\ } \ \tr(b^{2})\neq0,\ \tr(b)=0,\    \overline{\eta}(\tr(b^{2}))=-1.\\
                                                                        \end{array}
                                                                        \right.
$$
\end{Lem}

In order to determine the weight distribution of $ \C_{D} $ of \eqref{eq-a0} in Theorem \ref{Thm3}, we need  the next two lemmas.

\begin{Lem} [see \cite{DD14}, Lemma 9]\label{Lem16}
For each $ c \in \mathbb{F}_{p},$ set
$$\ u_{c}=|\{ x \in F_{q}: \tr(x^{2})=c \}|.$$
 If  $m $ is odd,
 then
 $$ u_{c}=p^{m-1}+p^{-1}\overline{\eta}(-1)\overline{\eta}(c)G\overline{G}. $$
\end{Lem}

\begin{Lem}\label{Lem17}
\ Let $ m $ be odd with $p\mid m.$ For each $ c\in \mathbb{F}_{p}^{*}, $ set
$$ v_{c}=|\{x\in F_{q}: \tr(x^{2})=c, \tr(x)=0\}|. $$
Then
$$ v_{c}=p^{m-2}+p^{-1}\overline{\eta}(-1)\overline{\eta}(c)G\overline{G}. $$
\end{Lem}
\begin{proof}
The proof of this lemma is similar to that of Lemma \ref{Lem10} and  we omit the details.
\end{proof}

Now we are ready to  prove  Theorem \ref{Thm3}.
Note that $ 2\nmid m$ and  $p\mid m. $ By Lemma \ref{Lem8}, we have $ n_{0}=p^{m-1}. $
It follows from \eqref{eq-3.1} and Lemma \ref{Lem15} that $\wt(c_b)=n_0-|N_b|$
\begin{align*}
 \in \left\{(p-1)p^{m-2}, (p-1)p^{m-2}\pm \frac{1}{p^{2}}\overline{\eta}(-1)G\overline{G}, (p-1)p^{m-2}\pm \frac{1}{p^{2}}\overline{\eta}(-1)(p-1)G\overline{G}\right\}.
 \end{align*}
Suppose
\begin{align*}
\omega_{1}&=(p-1)p^{m-2},\\
\omega_{2}&=(p-1)p^{m-2}+\frac{1}{p^{2}}\overline{\eta}(-1)G\overline{G},\\
\omega_{3}&=(p-1)p^{m-2}-\frac{1}{p^{2}}\overline{\eta}(-1)G\overline{G}, \\ \omega_{4}&=(p-1)p^{m-2}-\frac{1}{p^{2}}\overline{\eta}(-1)(p-1)G\overline{G},\\
\omega_{5}&=(p-1)p^{m-2}+\frac{1}{p^{2}}\overline{\eta}(-1)(p-1)G\overline{G}.
\end{align*}
By Lemmas \ref{Lem15}, \ref{Lem16} and \ref{Lem17}, we have $ A_{\omega_{1}}=p^{m-1}-1 $ and the following system of equations:
\begin{eqnarray}\label{eq-3.3}
\left\{\begin{array}{l}
         A_{w_2}+A_{w_4}=\frac{1}{2}(p-1)(p^{m-1}+p^{-1}\overline{\eta}(-1)G\overline{G}), \\
         A_{w_3}+A_{w_5}=\frac{1}{2}(p-1)(p^{m-1}-p^{-1}\overline{\eta}(-1)G\overline{G}), \\
         A_{w_4}=\frac{1}{2}(p-1)(p^{m-2}+p^{-1}\overline{\eta}(-1)G\overline{G}), \\
         A_{w_5}=\frac{1}{2}(p-1)(p^{m-2}-p^{-1}\overline{\eta}(-1)G\overline{G}).
       \end{array}
       \right.
\end{eqnarray}
Solving the system of equations of \eqref{eq-3.3} proves the weight distribution of \autoref{tal:weightdistribution3}.

\subsection{The second case of five-weight linear codes}
In this subsection, put $ 2\nmid m$ and  $p\nmid m. $  The last auxiliary result we need is the following.
\begin{Lem} \label{Lem18}
\ Let $ b\in \mathbb{F}_{q}^{*}$ and the symbols be the same as before. Then
$$
|N_{b}|=\left\{\begin{array}{ll}
                                                                          p^{m-2}, & \textrm{if\ }  \tr(b^{2})=0\textrm{\ and\ } \tr(b)\neq0, \\
                                                                          p^{m-2}+p^{-1}\overline{\eta}(-m)G\overline{G}, & \textrm{if\ } \tr(b^{2})=0\textrm{\ and\ }  \tr(b)=0  \\
                                                                          p^{m-2}-p^{-2}\overline{\eta}(-\tr(b^{2}))G\overline{G}, & \textrm{if\ }\tr(b^{2})\neq0,\ \tr(b)\neq0\textrm{\ and\ }    (\tr(b))^{2}\neq m\tr(b^{2}) \\
                                                                            & \textrm{or\ } \tr(b^{2})\neq0\textrm{\ and\ }  \tr(b)=0, \\
                                                                          p^{m-2}+p^{-2}\overline{\eta}(-m)(p-1)G\overline{G}, & \textrm{if\ }\tr(b^{2})\neq0,  \tr(b)\neq0\textrm{\ and\ }  (\tr(b))^{2}=m\tr(b^{2}).
                                                                        \end{array}
                                                                        \right.
$$
\end{Lem}
\begin{proof}
This lemma follows from  \eqref{eq-weight}, Lemmas \ref{Lem8} and Lemma \ref{Lem9},
\end{proof}

With the help of preceding lemmas we can now prove   Theorem \ref{Thm4}.
If  $ 2\nmid m$ and  $p\nmid m$,  by Lemma \ref{Lem8}, we have
$$ n_{0}=p^{m-1}+p^{-1}\overline{\eta}(-m)G\overline{G}. $$
 By Lemma \ref{Lem18}, we know $ wt(c_b)$ has five possible values. Let
\begin{align*}
w_1&=(p-1)p^{m-2}+\frac{1}{p}\overline{\eta}(-m)G\overline{G},\quad w_2=(p-1)p^{m-2},\\   w_3&=(p-1)p^{m-2}+\frac{1}{p^{2}}(p\overline{\eta}(-m)+ 1)G\overline{G}, \\ w_4&=(p-1)p^{m-2}+\frac{1}{p^{2}}(p\overline{\eta}(-m)- 1)G\overline{G}, \\ w_5&=(p-1)p^{m-2}+p^{-2}\overline{\eta}(-m)G\overline{G}.
\end{align*}
It follows from  Lemmas \ref{Lem11},  \ref{Lem12} and  \ref{Lem18} that
\begin{align*}
 A_{\omega_{1}}&=(p-1)(p^{m-2}-p^{-2}\overline{\eta}(-m)G\overline{G}), \\
 A_{\omega_{2}}&=p^{m-2}+p^{-2}\overline{\eta}(-m)(p-1)G\overline{G}-1, \\ A_{\omega_{5}}&=(p-1)p^{m-2}+p^{-2}\overline{\eta}(-m)(p-1)^{2}G\overline{G},
\end{align*}
where $A_{w_i}$ denotes the frequency of $w_i$.
It can be easily  checked  that the minimum distance of the dual code $ \mathcal{C}^{\perp}_{D} $ of $\mathcal{C}_{D}$ is  equal to $ 2 $.  By the first two Pless Power Moments (\cite{HP}, p. 260) the frequency $A_{w_i}$ of $w_{i}$ satisfies the following equations:
\begin{eqnarray}\label{eq-a1}
\left\{\begin{array}{l}
         A_{w_1}+A_{w_2}+A_{w_3}+A_{w_4}+A_{w_5}=p^{m}-1, \\
         w_1A_{w_1}+w_2A_{w_2}+w_3A_{w_3}+w_4A_{w_4}+w_5A_{w_5}=p^{m-1}(p-1)n,
       \end{array}
       \right.
\end{eqnarray}
where $n=p^{m-1}+p^{-1}\overline{\eta}(-m)G\overline{G}-1$.
 A simple manipulation leads to the weight distribution of \autoref{tal:weightdistribution4}.

\section{Concluding Remarks}
In this paper, we present a class of three-weight and five-weight linear codes. There is a survey on three-weight codes in \cite{DLLZ}. A number of three-weight and five-weight codes were discussed in \cite{CJ,Ding09,DGZ,DD14,DD15,X,ZD,ZLFH,ZDLZ}.

Let $w_{\min}$ and $w_{\max}$ denote the minimum and maximum nonzero weight of a linear code $\C$. The linear code $\C$ with
$w_{\min}/w_{\max}>(p-1)/p$ can be used to construct a secret sharing scheme with  interesting access structures (see \cite{YD}).

Let $m\geq4$. Then for  the linear code $\C_D$ of Theorem \ref{Thm1}, we have
$$
\frac{w_{\min}}{w_{\max}}=\frac{(p-1)p^{m-2}-(p-2)p^{\frac{m-2}{2}}}{(p-1)p^{m-2}} \textrm{\ or\ }
\frac{(p-1)p^{m-2}}{(p-1)p^{m-2}+(p-2)p^{\frac{m-2}{2}}} .
$$
It can be easily checked  that
$$
\frac{(p-1)p^{m-2}}{(p-1)p^{m-2}+(p-2)p^{\frac{m-2}{2}}}>\frac{(p-1)p^{m-2}-(p-2)p^{\frac{m-2}{2}}}{(p-1)p^{m-2}}>\frac{p-1}{p}.
$$
Hence,
$$\frac{w_{\min}}{w_{\max}}>\frac{p-1}{p}.$$

Let $m\geq6$. Then for the linear code $\C_D$ of Theorem \ref{Thm2}, we have
$$
\frac{w_{\min}}{w_{\max}}=\frac{(p-1)p^{m-2}-2p^{\frac{m-2}{2}}}{(p-1)p^{m-2}}\textrm{\ or\ }
 \frac{(p-1)p^{m-2}}{(p-1)p^{m-2}+2p^{\frac{m-2}{2}}}.
$$
Simple computation shows that
$$
\frac{(p-1)p^{m-2}}{(p-1)p^{m-2}+2p^{\frac{m-2}{2}}}>\frac{(p-1)p^{m-2}-2p^{\frac{m-2}{2}}}{(p-1)p^{m-2}}>\frac{p-1}{p}.
$$
Therefore,
$$\frac{w_{\min}}{w_{\max}}>\frac{p-1}{p}.$$

Let $m\geq5$. Then for the linear code $\C_D$ of Theorem \ref{Thm3}, we have
$$
\frac{w_{\min}}{w_{\max}}=\frac{(p-1)p^{m-2}-(p-1)p^{\frac{m-3}{2}}}{(p-1)p^{m-2}+(p-1)p^{\frac{m-3}{2}}}>\frac{p-1}{p}.
$$

Let $m\geq5$. Then for the linear code $\C_D$ of Theorem \ref{Thm4}, we have
$$
\frac{w_{\min}}{w_{\max}}=\frac{(p-1)p^{m-2}-(p+1)p^{\frac{m-3}{2}}}{(p-1)p^{m-2}}
 \textrm{\ or\ }\frac{(p-1)p^{m-2}}{(p-1)p^{m-2}+(p+1)p^{\frac{m-3}{2}}}.
$$
It is easy to show that $$\frac{(p-1)p^{m-2}}{(p-1)p^{m-2}+(p+1)p^{\frac{m-3}{2}}}>\frac{(p-1)p^{m-2}-(p+1)p^{\frac{m-3}{2}}}{(p-1)p^{m-2}}>\frac{p-1}{p}.$$
Then we get
$$
\frac{w_{\min}}{w_{\max}}>\frac{p-1}{p}.
$$
To sum up, the linear codes $\C_D$ with $m\geq5$ can be employed to get secret sharing schemes.

{\bf Acknowledgements. } \quad \ This research is supported by a National Key Basic Research Project of China (2011CB302400), National Natural Science Foundation of China (61379139),  the ``Strategic Priority Research Program" of the Chinese Academy of Sciences, Grant No. XDA06010701 and Foundation of NSSFC(No.13CTJ006).

\end{document}